\documentclass[showpacs,amsmath,amssymb,onecolumn,pra]{revtex4-1}

\usepackage[dvips]{graphicx} 
\usepackage{amsfonts}
\usepackage{amssymb}
\usepackage{amscd}
\usepackage{amsmath}    
\usepackage{enumerate}
\usepackage{epsfig}
\usepackage{subfigure}
\usepackage{xcolor}
\usepackage{amsthm}
\usepackage{framed}
\usepackage{multirow}

\newtheorem{theorem}{Theorem}
\newtheorem{lemma}{Lemma}
\newtheorem{corollary}{Corollary}

\newtheorem{definition}{Definition}

\newtheorem{observation}{Observation}

\newcommand{\bra}[1]{\mbox{$\left\langle #1 \right|$}}
\newcommand{\ket}[1]{\mbox{$\left| #1 \right\rangle$}}

\newcommand{\comments}[1]{}

\begin{document}
\preprint{APS/123-QED}
\title{Decomposition of a symmetric multipartite observable}

\date{\today}
\author{You Zhou$^{\dag}$, Chenghao Guo$^{\dag}$, Xiongfeng Ma}
\email{xma@tsinghua.edu.cn}
\affiliation{Center for Quantum Information, Institute for Interdisciplinary Information Sciences, Tsinghua University, Beijing 100084, China}

\begin{abstract}
Production and verification of multipartite quantum state are an essential step in quantum information processing. In this work, we propose an efficient method to decompose symmetric multipartite observables, which are invariant under permutations between parties, with only $(N+1)(N+2)/2$ local measurement settings, where $N$ is the number of qubits. We apply the decomposition technique to evaluate the fidelity between an unknown prepared state and any target permutation invariant state. In addition, for some typical permutation invariant states, such as the Dicke state with a constant number of excitations, $m$, we derive a tight linear bound on the number of local measurement settings, $m(2m+3)N+1$. Meanwhile, for the $GHZ$ state, the $W$ state, and the Dicke state, we prove a linear lower bound, $\Theta(N)$. Hence, for these particular states, our decomposition technique is optimal.
\end{abstract}

\maketitle
\section{Introduction}
Quantum states with genuine multipartite entanglement, such as the GHZ state \cite{greenberger1989going}, the Dicke state (including the $W$ state) \cite{Dicke1954Coherence,Dur2000three}, and the general graph (stabilizer) state \cite{Hein2004Multiparty,gottesman1997stabilizer}, are essential ingredients for many quantum information processing tasks, such as multipartite
quantum key distribution \cite{Chen2007Multi}, quantum secret sharing \cite{Hillery1999secret,Cleve1999Share}, quantum error correction \cite{gottesman1997stabilizer,nielsen2010quantum}, measurement-based quantum computing \cite{Raussendorf2001One}, and quantum metrology \cite{Wineland1994Squeezed,Giovannetti2004measure}. In practice, due to the noise caused by the uncontrolled interaction between the system and environment, the prepared state unavoidably deviates from the target one. Hence, it is necessary to quantify such deviation, which acts as a calibration for the experimental system and provides the basis of further information processing.

A straightforward method to benchmark the system is quantum state tomography \cite{Vogel1989Determination,Paris2004esimation}. In reality, due to the tensor product structure of the Hilbert space, the required resources scale exponentially with the number of system parties (say, qubits) in tomography. In the last decade, the qubit number under manipulation increases significantly in various experiment systems, such as those based on ion-trap \cite{Monz2011qubit}, superconducting \cite{Song2017Sconducting}, and linear optics \cite{Wang2016Photon}. Thus, it is impractical to directly conduct tomography for state-of-the-art multipartite quantum systems. Fortunately, the required resources can be dramatically reduced, if one possesses some preknowledge about the prepared state and takes advantage of symmetries. In this spirit, several efficient tomography methods were put forward for various types of quantum states, such as the low-rank state \cite{Gross2010Tomography,Flammia2012compressed}, the matrix product state \cite{Cramer2010Efficient,Baumgratz2013Scalable,Lanyon2017Efficient}, and the permutation invariant (PI) state \cite{Toth2010Permutation,Tobias2012Permutation}. The insight underlying this simplification is that one only needs the parameters of the ansatz states there, such as the tensor network state and the PI state, whose number only increases polynomially with the number of qubits $N$.


In practical quantum information tasks, instead of gaining all the information about the density matrix, one only needs to guarantee that the prepared state holds sufficiently high fidelity with the target state. If focusing on fidelity evaluation instead of a full tomography, one can further reduce the measurement efforts. Sometimes, one only needs to detect or witness entanglement for multipartite systems. Since quantum states with high symmetries are widely used in information processing, these tasks generally involve symmetric observables, which are invariant under permutations of parties. In addition, multipartite measurement is normally very challenging in practice. Instead, it is often broken down to local measurements \cite{TERHAL2002Detecting,Bourennane2004Experimental}. Take an $N$-qubit system for example, to measure the fidelity between a prepared state $\rho$ and the GHZ state, $\mathrm{Tr}(\rho\ket{GHZ}\bra{GHZ})$, one cannot measure it directly with $\ket{GHZ}\bra{GHZ}$. Instead, $\langle\ket{GHZ}\bra{GHZ}\rangle$ is broken down to a set of local measurements $\{A^{\otimes{N}}\}$, whose number determines the complexity of fidelity evaluation. Note that from a local measurement setting (LMS) $A^{\otimes{N}}$, not only the expectation value $\langle A^{\otimes{N}} \rangle$, but also the full statistics can be obtained. For instance, one can get the probability of any specific measurement result, say an $N$-bit string, from the Pauli-$Z$ measurement, $\sigma_Z^{\otimes{N}}$.

In this work, we focus on estimating the fidelity with PI states. The PI state set $\mathcal{S}_{PI}$ contains all the states which are invariant under any subsystem permutation,
\begin{equation}\label{Eq:defPI}
\begin{aligned}
\rho^{PI}=P(\pi)\rho^{PI}P(\pi), \; \forall \pi
\end{aligned}
\end{equation}
where $P(\pi)$ is permutation operation of the element $\pi$ in the symmetry group $S_N$.
It is worth mentioning another related state set, the symmetric state set $\mathcal{S}_S$, which contains all the pure states satisfying,
\begin{equation}\label{Eq:defSym}
\begin{aligned}
\ket{\psi_s}=P(\pi)\ket{\psi_s}, \; \forall \pi
\end{aligned}
\end{equation}
and their convex combination.
It is not hard to see that $\mathcal{S}_S\subset \mathcal{S}_{PI}$
\footnote{
For example, the Bell state $\ket{\Psi^-}=\frac{1}{\sqrt{2}}(\ket{01}-\ket{10})$ has eigenvalue $-1$ of the permutation operation between the two qubits, i.e., $P(2,1)\ket{\Psi^-}=-\ket{\Psi^-}$. Hence $\ket{\Psi^-} \notin  S_{S}$ but $\ket{\Psi^-}\in S_{PI}$. Symmetric states possess many interesting properties. For example, there is a dichotomy for the entanglement
depth of the pure symmetric state \cite{Ichikawa2008Exchange,Exchange2010Wei}, i.e., they are either product states or fully entangled.
}.

By decomposing a symmetric observable, we construct a set of $(N+1)(N+2)/2$ LMSs to evaluate the fidelity between a prepared state and any target PI state. To do this, we first give a general theorem which constructs a set of product-state basis for the symmetric subspace of the $N$-qudit Hilbert space. Then, based on this theorem, one can find a set of product operators as the basis for the symmetric subspace of the $N$-partite operator space, which contains all PI states. Thus, any PI state can be decomposed as the linear combination of these basis operators. As a result, by measuring all the basis operators with local measurements, one can finally evaluate the fidelity with respect to any target PI state by only post-processing the measurement results. By observing that several product operators can be measured with the same LMS, we finally reduce the number of LMSs to $(N+1)(N+2)/2$.

Moreover, based on this general decomposition, we can further reduce the number of LMSs for some typical PI states, such as the $GHZ$ state, the $W$ state, and the Dicke state with a constant number of excitations. Define measurement complexity as the minimal number of LMSs to decompose a state and we systemically study their measurement complexities. To give an upper bound of the measurement complexity, we show an explicit decomposition of the state by finding a smaller subspace of the symmetric subspace which contains the state and constructing the basis of it with fewer LMSs. To give a lower bound, we find some subspace where the projection of the target state shows some specific form, and prove that one needs enough LMSs to reestablish the same form. Combing the upper and the lower bounds, we finally show that their measurement complexities all scale as $\Theta(N)$.

The paper is organized as follows. In Sec.~\ref{Sec:Sym}, we construct a set of product-state bases for the symmetric subspace. In Sec.~\ref{sec:PI}, we propose a method based on the previous theorem, with $(N+1)(N+2)/2$ LMSs, to evaluate the fidelity between an unknown state and any target PI state. In Sec.~\ref{sec:typical}, we further reduce the number of LMSs in fidelity evaluation for some special PI states, $GHZ$, $W$, and Dicke. We finally conclude in Sec.~\ref{sec:final} with further discussions.

\section{Product-state basis for symmetric subspace}\label{Sec:Sym}
In this section, we construct a set of linearly independent vectors (states) in the product form, which can span the symmetric subspace. This construction will help us to find LMSs for decomposition of symmetric operators and fidelity evaluation of the PI state in the following sections.

First, let us briefly review the symmetric subspace of an $N$-qudit Hilbert space $\mathcal{H}_d^{\otimes{N}}$, denoted by $\mathrm{Sym}_N(\mathcal{H}_d)$. Given an element $\pi$ in the symmetric group $S_N$ with $N$ letters, the corresponding permutation operator defined on $\mathcal{H}_d^{\otimes{N}}$ is,
\begin{equation}\label{}
\begin{aligned}
P_d(\pi)=\sum_{i_1,\cdots ,i_N\in [d]}\ket{i_{\pi^{-1}(1)},\cdots ,i_{\pi^{-1}(N)}}\bra{i_1,\cdots,i_N},
\end{aligned}
\end{equation}
where $\{\ket{0},\ket{1},\cdots,\ket{d-1}\}$ is the local basis for each qudit and $[d]=\{0,1,\cdots,d-1\}$. The symmetric subspace
$ \mathrm{Sym}_N(\mathcal{H}_d)\subseteq \mathcal{H}_d^{\otimes{N}}$ contains all the pure states which are invariant under permutation,
\begin{equation}\label{}
\mathrm{Sym}_N(\mathcal{H}_d)=\{\ket{\Psi}\in \mathcal{H}_d^{\otimes{N}}: P_d(\pi)\ket{\Psi}=\ket{\Psi},\;\; \forall \pi\in S_N\}.
\end{equation}
In the qubit case with $d=2$, we denote the permutation operator $P(\pi)$ for simplicity.

It is known that the dimension of  $\mathrm{Sym}_N(\mathcal{H}_d)$ is given by \cite{Harrow2013symmetric}
\begin{equation}\label{Eq:symspacesize}
\begin{aligned}
D_{S}=\binom{N+d-1}{N}=\frac{(N+d-1)!}{N!(d-1)!},
\end{aligned}
\end{equation}
and there is a set of orthogonal (unnormalized) basis of $\mathrm{Sym}_N(\mathcal{H}_d)$,
\begin{equation}\label{Eq:Dbasisd}
\begin{aligned}
\left\{\ket{\Psi_{\vec{i}}}=\sum_{\pi}\ket{0}^{\otimes i_0}\ket{1}^{\otimes i_1}\cdots \ket{d-1}^{\otimes i_{d-1}} \bigg| i_k\in \mathbb{Z}^+,\;\; \sum_{k=0}^{d-1} i_k=N\right\},
\end{aligned}
\end{equation}
where $\mathbb{Z}^+$ denotes the nonnegative integer set and $\vec{i}=(i_0,i_1,\cdots,i_{d-1})$ is a $d$-dimensional vector. Here, $\sum_{\pi}$ represents the summation over all permutations of $N$ qudits that yield different expressions. 
We keep this notation throughout the paper.

Meanwhile, the symmetric subspace $\mathrm{Sym}_N(\mathcal{H}_d)$ can be spanned by the symmetric product states,
\begin{equation}\label{Eq:Sym:prod}
\mathrm{Sym}_N(\mathcal{H}_d)=\mathrm{span}\{\ket{\phi}^{\otimes N}:\ket{\phi}\in \mathcal{H}_d\}.
\end{equation}
For a finite $N$, it is not hard to see that the symmetric product states, $\ket{\phi}^{\otimes N}$, are linearly dependent. In the following, we construct a product-state basis for the symmetric subspace $\mathrm{Sym}_N(\mathcal{H}_d)$, by selecting $\binom{N+d-1}{N}$ linearly independent product states, as shown in Theorem \ref{Th:basis}.

Define a $d\times (N+1)$  matrix with complex elements $a_{k,j}$ satisfying
\begin{equation}\label{Eq:Adn}
\begin{aligned}
&a_{0,j}=1,\\
&a_{k,j}\neq a_{k,j'}, \forall 1\le k\le d-1,
\end{aligned}
\end{equation}
for all $0\leq j\neq j'\leq N$. That is, all the elements in the zeroth row are $1$; and for the other rows, the elements are different for different columns.

\begin{theorem}\label{Th:basis}
The following state set $\mathcal{B}$ contains $\binom{N+d-1}{N}$ linearly independent vectors which are (unnormalized) symmetric product states, and they can span the symmetric subspace $\mathrm{Sym}_N(\mathcal{H}_d)$,
\begin{equation}\label{}
\begin{aligned}
\mathcal{B}=\left\{\ket{\Phi_{\vec{j}}}=\Big(a_{0,j_0}\ket{0}+a_{1,j_1}\ket{1}+\cdots+a_{d-1,j_{d-1}}\ket{d-1}\Big)^{\otimes N}\bigg|j_k\in\mathbb{Z}^+,\; \sum_{k=0}^{d-1} j_k=N\right\},
\end{aligned}
\end{equation}
where the coefficients $a_{k,j_k}$ are selected from any matrix satisfying Eq.~\eqref{Eq:Adn}.
\end{theorem}

In the following, we call the vectors in $\mathcal{B}$ a set of basis of $\mathrm{Sym}_N(\mathcal{H}_d)$, even though they might not be  orthogonal with each other.
The complete proof of Theorem \ref{Th:basis} is shown in Appendix \ref{prf:thm1}, which is based on induction. Here, we present the qubit case of Theorem \ref{Th:basis} in Corollary \ref{Co:symqubit} and provide a simple proof.

\begin{corollary}\label{Co:symqubit}
The following state set $\mathcal{B}$ contains $N+1$ linearly independent vectors which are (unnormalized) symmetric product states, and they can span the symmetric subspace $\mathrm{Sym}_N(\mathcal{H}_2)$,
\begin{equation}\label{eq:symmqubitset}
\begin{aligned}
\mathcal{B}=\left\{\ket{\Phi_j}=(\ket{0}+a_{j}\ket{1})^{\otimes N}|0\le j\le N\right\},
\end{aligned}
\end{equation}
where $a_i$ are complex numbers and $a_j\neq a_{j'}$ for $j\neq j'$.
\end{corollary}

\begin{proof}
The state set $\mathcal{B}$ only contains symmetric product states in $\mathrm{Sym}_N(\mathcal{H}_2)$. From Eq.~\eqref{Eq:symspacesize}, the dimension of $\mathrm{Sym}_N(\mathbb{C}^2)$ is $N+1$, which equals the cardinality of $\mathcal{B}$. Therefore, we only need to prove that the states in $\mathcal{B}$ are linearly independent.

In this qubit case, the orthogonal basis, given in Eq.~\eqref{Eq:Dbasisd},
\begin{equation}\label{Eq:SymBasis2}
\begin{aligned}
\left\{\ket{\Psi_i}=\sum_{\pi}\ket{0}^{\otimes N-i}\ket{1}^{\otimes i},0\le i\le N\right\},
\end{aligned}
\end{equation}
contains all the Dicke states. If one expands $\ket{\Phi_{j}}$ of Eq.~\eqref{eq:symmqubitset} in this $\{\ket{\Psi_i}\}$ basis, the coefficients are $(1,a_j,a_j^2\cdots,a_j^N)^T$. The matrix formed by the coefficients $(\ket{\Phi_0},\ket{\Phi_1},\cdots, \ket{\Phi_N})$ is a Vandermonde matrix which is nonsingular. Consequently, $\ket{\Phi_k}$ are linearly independent and form a basis of $\mathrm{Sym}_N(\mathcal{H}_2)$.
\end{proof}

\section{Symmetric observable decomposition and fidelity evaluation}\label{sec:PI}
In this section, we propose a method to decompose a symmetric observable and apply it to evaluate the fidelity between a prepared state $\rho$ and any target PI state using $(N+1)(N+2)/2$ LMSs. Here, we only focus on the $N$-qubit scenario, but the method can be generalized to the $N$-qudit case. As shown in Eq.~\eqref{Eq:defPI}, a PI state is defined on the density matrix level. That is, $\rho^{PI}$ is invariant under any permutation operation among qubits.
Due to this permutation invariant property, we only need to consider the case where the local operators in LMS are the same for all qubits \cite{Toth2010Permutation}, that is, in the form of $A^{\otimes{N}}$, where $A$ is a qubit Hermitian operator. Generally, one needs a set of LMSs $\{A^{\otimes{N}}_i\}$ to evaluate the fidelity. Also, the target state is normally pure. We have the following theorem.

\begin{theorem}\label{Th:PI:Fid}
For any $N$-qubit target PI state $\ket{\Psi^{PI}}$, the fidelity between a prepared state $\rho$ and $\ket{\Psi^{PI}}$,
\begin{equation}\label{Eq:fid}
F=\bra{\Psi^{PI}}\rho\ket{\Psi^{PI}},
\end{equation}
can be evaluated with $(N+1)(N+2)/2$ LMSs. 
\end{theorem}

Denote the projector formed by the PI state, $\ket{\Psi^{PI}}\bra{\Psi^{PI}}$, to be $\Psi^{PI}$. In order to measure the fidelity $F$ in Eq.~\eqref{Eq:fid}, one should decompose the projector $\Psi^{PI}$ into local measurements. In the following, we introduce the symmetric subspace of $N$-qubit Hermitian operators where $\Psi^{PI}$ is located. Then, by constructing a set of tensor-product bases of this symmetric subspace, we can accordingly decompose $\Psi^{PI}$.

\begin{proof}
Let us first define the symmetric subspace of $N$-qubit Hermitian operators. Denote the $N$-qubit Pauli group to be $G_N$, whose element, called the $N$-qubit Pauli operator, is a tensor product of single qubit Pauli operators and identity $G_1=\{\mathbb{I}, \sigma_X, \sigma_Y, \sigma_Z\}$. An $N$-qubit Hermitian operator $M$ can be written as the linear combination of the Pauli operators in $G_N$. Since the operators we consider here are Hermitian, their coefficients must be real. Thus, the operator space $F_{G_N}$ is isomorphic to $(\mathbb{R}^4)^{\otimes N}$, with $\mathbb{R}$ denoting the real domain. Then, the corresponding symmetric subspace, $\mathrm{Sym}_N(G_1)$, is defined as
\begin{equation}\label{}
\begin{aligned}
\mathrm{Sym}_N(G_1)=\left\{M\in F_{G_N}:P(\pi)MP(\pi)=M,\forall\pi\in S_N\right\}.
\end{aligned}
\end{equation}
By the definition in Eq.~\eqref{Eq:defPI}, any PI state, $\rho^{PI}\in \mathrm{Sym}_N(G_1)$.

Since $\mathrm{Sym}_N(G_1)$ is isomorphic to $\mathrm{Sym}_N(\mathbb{R}^4)$, the dimension of $\mathrm{Sym}_N(G_1)$ is $\binom{N+3}{3}$ according to Eq.~\eqref{Eq:symspacesize}. Meanwhile, the results shown in Sec.~\ref{Sec:Sym} can be directly applied to the operator space here.
\begin{enumerate}
\item
Similarly to the orthogonal basis shown in Eq.~\eqref{Eq:Dbasisd}, the following Hermitian operators form an orthogonal basis (in the sense of the Hilbert-Schmidt inner product) of $\mathrm{Sym}_N(G_1)$,
\begin{equation}\label{Eq:Mijk}
\begin{aligned}
M_{i,j,k}
=\sum_{\pi}\mathbb{I}^{\otimes i}\otimes \sigma_X^{\otimes j}\otimes\sigma_Y^{\otimes k}\otimes\sigma_Z^{\otimes (N-i-j-k)},
\end{aligned}
\end{equation}
where $\sigma_X^{\otimes j}$ denotes that there are totally $j$ qubits with $\sigma_X$ on them, similar for $\mathbb{I}^{\otimes i}, \sigma_Y^{\otimes k}, \sigma_Z^{\otimes (N-i-j-k)}$.

\item
Similarly to Eq.~\eqref{Eq:Sym:prod}, product operators can also span the symmetric subspace,
\begin{equation}\label{}
\mathrm{Sym}_N(G_1)=\mathrm{span}\{A^{\otimes N}:A=a\mathbb{I}+b\sigma_X+c\sigma_Y+d\sigma_Z\},
\end{equation}
where $a,b,c,d\in \mathbb{R}$.

\item
Similarly to the proof of Theorem \ref{Th:basis}, with the previous two steps, we only need to select $\binom{N+3}{3}$ linearly independent operators $\{A^{\otimes N}\}$ to act as the product-form basis of $\mathrm{Sym}_N(G_1)$. According to Theorem \ref{Th:basis}, we can construct a basis set for $\mathrm{Sym}_N(G_1)$ with product operators, where the local bases are $\{\mathbb{I}, \sigma_X, \sigma_Y, \sigma_Z\}$. To be specific,
\begin{equation}\label{Eq:Oset1}
\mathcal{B}_o=\left\{(a_{1,j_1}\mathbb{I}+a_{2,j_2}\sigma_X+a_{3,j_3}\sigma_Y+a_{0,j_0}\sigma_Z)^{\otimes N}\bigg|\sum_{k=0}^{3} j_k=N\right\},
\end{equation}
where the subscript ``o'' denotes an operator set, and the coefficients $a_{k,j_k}$ with $0\leq k \leq 3$ and $0\leq j_k \leq N$ are real numbers satisfying Eq.~\eqref{Eq:Adn}.
\end{enumerate}
After step 3, we further reduce the number $\binom{N+3}{3}$ down to $\binom{N+2}{2}$ due to the fact that some of the product operators can be measured by the same LMS.
Since $a_{0,j_0}=1$ for any $j_0$, the operator set in Eq.~\eqref{Eq:Oset1} can also be written as
\begin{equation}\label{Eq:basisOp1}
\begin{aligned}
\mathcal{B}_o=\left\{(a_i\mathbb{I}+b_j\sigma_X+c_k\sigma_Y+\sigma_Z)^{\otimes N}\bigg|0\le i,j,k\le N,\; i+j+k\le N\right\},
\end{aligned}
\end{equation}
where, for simplicity of notation, let $a_i=a_{1,i}$, $b_j=a_{2,j}$, and $c_k=a_{3,k}$.

On account of Observation \ref{Ob:reduce} below, product operators in $\mathcal{B}_o$ with different $a_i$ but the same $b_j$ and $c_k$ can be obtained via the same LMS $(b_j\sigma_X+c_k\sigma_Y+\sigma_Z)^{\otimes N}$. Therefore, the total number of LMSs required is the number of different parameter sets $(b_j,c_k)$, which equals to the number of solutions for $j+k\le N$, i.e., $\binom{N+2}{2}=(N+1)(N+2)/2$.

Consequently, one can utilize $(N+1)(N+2)/2$ LMSs to obtain the expectation values of all the basis operators in $\mathcal{B}_o$ of symmetric subspace $\mathrm{Sym}_N(G_1)$. Since any PI state $\Psi^{PI}\in \mathrm{Sym}_N(G_1)$, we can decompose any $\Psi^{PI}$ in this basis and finally obtain the fidelity in Eq.~\eqref{Eq:fid}.
\end{proof}

\begin{observation}\label{Ob:reduce}
The expectation value of the product operator $(a\mathbb{I}+b\sigma_X+c\sigma_Y+d\sigma_Z)^{\otimes N}$ can be obtained via the LMS $(b\sigma_X+c\sigma_Y+d\sigma_Z)^{\otimes N}$.
\end{observation}
\begin{proof}
\begin{equation}\label{Eq:reduce}
\begin{aligned}
&(a\mathbb{I}+b\sigma_X+c\sigma_Y+d\sigma_Z)^{\otimes N}\\=&\sum_{i=0}^Na^i\sum_{\pi}\mathbb{I}^{\otimes i}\otimes (b\sigma_X+c\sigma_Y+d\sigma_Z)^{\otimes N-i}.
\end{aligned}
\end{equation}
Note that each term in Eq.~\eqref{Eq:reduce} can be directly obtained from the LMS, $(b\sigma_X+c\sigma_Y+d\sigma_Z)^{\otimes N}$.
\end{proof}

Some remarks on and discussion of Theorem \ref{Th:PI:Fid} are listed as follows.

\begin{enumerate}
\item
We show how to efficiently decompose a general PI state in a set of product-form basis in Appendix \ref{Ap:decomp}, which is useful in practical implementation. And we remark that this decomposition is suitable for general symmetric operators, which may be useful for other related problems. In fact, on account of Theorem \ref{Th:basis}, it is possible to choose other sets of product-form bases. Thus, we also discuss how to choose a basis which is more robust to noise there.

\item
Theorem \ref{Th:PI:Fid} can be directly extended to the qudit case, by considering the generalized local Pauli basis for qudits.


\item
Our method can evaluate the fidelity between an unknown prepared state and any PI state with the same $\binom{N+2}{2}$ LMSs, by only postprocessing the measurement results.

\item
The basis operators in $\mathcal{B}_o$ of Eq.~\eqref{Eq:Oset1} span the symmetric subspace $\mathrm{Sym}_N(G_1)$, which contains all the PI states, $\rho^{PI}\in \mathrm{Sym}_N(G_1)$. From the proof of Theorem \ref{Th:PI:Fid}, the expectation values of all these operators in $\mathcal{B}_o$ can be measured with $\binom{N+2}{2}$ LMSs. As a result, in addition to the fidelity evaluation in Eq.~\eqref{Eq:fid}, one can also gain more detailed information on a generic PI state $\rho^{PI}$ with these measurements. This is related to permutation invariant tomography \cite{Toth2010Permutation}, where one can extract the permutation invariant component of a state with $\binom{N+2}{2}$ LMSs. Note that our basis constructed in Eq.~\eqref{Eq:Oset1} is more explicit and different, compared to the one shown in Ref.~\cite{Toth2010Permutation}.
\end{enumerate}

We show in the following that our decomposition method can further help to reduce the number of LMSs for some special PI states.
To be specific, it is known that for the GHZ state and the W state, one only needs $N+1$ and $2N-1$ LMSs to decompose them respectively \cite{Guhne2007Toolbox}, compared with $\binom{N+2}{2}$ for a general PI state. In the following sections, we define a quantity called the \emph{measurement complexity} which quantifies the minimal number of LMSs to decompose a state, and systemically study the measurement complexities for the GHZ state, the W state, and the Dicke state.
\section{Complexity upper bound of the Dicke State}
In the previous section, we have obtained a method to decompose any PI state, and in the following we focus on reducing the number of LMSs for some specific PI states. Certain PI states are typical for quantum information processing and have been extensively studied, including the GHZ state, the W state, and the Dicke state. In this section and the next section, the number of LMSs required for these states are discussed.


First, let us give the definition of the \emph{measurement complexity} of PI states, strictly speaking, symmetric-measurement complexity, since the LMSs utilized here are in the symmetric form. For simplicity, we use the term measurement complexity without confusion.
\begin{definition}\label{def:MC}
For an $N$-qubit PI state $\rho$, measurement complexity $C_S(\rho)$ is the minimal number of LMSs to decompose it,
\begin{equation}\label{Eq:MCDef}
\begin{aligned}
&C_S(\rho)=\min n_A, \\
&s.t., \rho=\sum_{i=1}^{n_A}\sum_{j=1}^N\alpha_{ij}\sum_{\pi}\mathbb{I}^{\otimes j}A_i^{\otimes N-j},
\end{aligned}
\end{equation}
where $A_i=b_i\sigma_X+c_i\sigma_Y+d_i\sigma_Z$ and $\alpha_{ij}, b_i, c_i, d_i \in \mathbb{R}$.
\end{definition}

Here we allow $\sum_{\pi}\mathbb{I}^{\otimes j}A_i^{\otimes N-j}$ to appear in the summation in Eq.~\eqref{Eq:MCDef}, since its expectation value can be inferred from the LMS $A_i^{\otimes N}$. Note that we do not need to introduce $\mathbb{I}$ into $A_i$, like $(a_i\mathbb{I}+b_i\sigma_X+c_i\sigma_Y+d_i\sigma_Z)^{\otimes N}$, since any operator $\sum_{\pi}\mathbb{I}^{\otimes j}(a_i\mathbb{I}+A_i)^{\otimes N-j}$ can be written as a combination of $\sum_{\pi}\mathbb{I}^{\otimes j}A_i^{\otimes N-j}$.

A summary of the results is reported in Table \ref{table:MC}. It is noteworthy that the measurement complexities for these states all increase linearly with the number of qubits, $\Theta(N)$, whereas the measurement complexity of a general PI state increases quadratically with $N$, as stated in Theorem \ref{Th:PI:Fid}.
\begin{table*}[tbph]
\caption{Measurement complexity of some typical PI states. For these states, the measurement complexity is linear, $\Theta(N)$.}\label{table:MC}
\centering
\begin{tabular}{|c|c|c|}
\hline
State & Upper bound & Lower bound \\
\hline
W & $2N-1$ \cite{Guhne2007Toolbox} & $N-1$ \\
Dicke & $m(2m+3)N+1$ & $N-2m+1$ \\
GHZ & $N+1$ \cite{Guhne2007Toolbox} & $\lceil\frac{N+1}{2}\rceil$ \\
\hline
\end{tabular}
\end{table*}

In this section, we focus on finding the upper bound of the measurement complexity. To find a upper bound, one needs to give a specific decomposition of the state. Note that one can use $N+1$ and $2N-1$ LMSs to decompose the GHZ and the W state respectively, which are listed as the upper bounds in Table \ref{table:MC}.

In the following, we give an explicit decomposition of the Dicke state $\ket{D_{N,m}}$
of $N$ qubits with $m$ excitations \cite{Dicke1954Coherence},
\begin{equation}\label{}
\begin{aligned}
\ket{D_{N,m}}=\frac{1}{\sqrt{\binom{N}{m}}}\sum_{\pi}\ket{0}^{\otimes N-m}\ket{1}^{\otimes m}.
\end{aligned}
\end{equation}
using $m(2m+3)N+1$ LMSs. Note that as $m=1$, it is the W state \cite{Dur2000three},
\begin{equation}\label{}
\begin{aligned}
\ket{W_N}=\frac{1}{\sqrt{N}}\left(\ket{10\cdots 0 1}+\ket{01\cdots 0}+\cdots +\ket{00\cdots 1}\right).
\end{aligned}
\end{equation}

To do so, we characterize a subspace of the symmetric subspace that contains the state, then construct a basis set for that subspace with some LMSs.
Consequently, one obtains an upper bound of the measurement complexity of $\ket{D_{N,m}}$, which is summarized in the following theorem.
\begin{theorem}
The measurement complexity of the Dicke state $\ket{D_{N,m}}$ is upper bounded by
\begin{equation}\label{}
\begin{aligned}
C_S(D_{N,m})\le m(2m+3)N+1.
\end{aligned}
\end{equation}
\end{theorem}

\begin{proof}
The density matrix of $\ket{D_{N,m}}$ can be explicitly written as
\begin{equation}\label{Eq:DecomDicke}
\begin{aligned}
D_{N,m}&= \frac{1}{\binom{N}{m}}\sum_{t=0}^{m}\Theta_t, \\
\Theta_t &\equiv \sum_{\pi}(\ket{1}\bra{1})^{\otimes m-t}\otimes (\ket{0}\bra{1})^{\otimes t}\otimes(\ket{1}\bra{0})^{\otimes t}\otimes (\ket{0}\bra{0})^{\otimes N-m-t}, \\
\end{aligned}
\end{equation}
where we denote each term in the summation $\Theta_t$, and $\Theta_0$ is the diagonal term. It is clear that $\Theta_0$ can be obtained by the $Z$-basis measurement $\sigma_Z^{\otimes N}$. Thus in the following we focus on $\Theta_t$ with $1\leq t \leq m$, which are the off-diagonal terms.

With the relations $\ket{1}\bra{1}=(\mathbb{I}-\sigma_Z)/2$, $\ket{0}\bra{0}=(\mathbb{I}+\sigma_Z)/2$, $\ket{0}\bra{1}=(\sigma_X+i\sigma_Y)/2$ and $\ket{1}\bra{0}=(\sigma_X-i\sigma_Y)/2$, $\Theta_t$ can be written in the form
\begin{equation}\label{Eq:Theta}
\begin{aligned}
\Theta_t=\sum_{\pi}(\frac{\mathbb{I}-\sigma_Z}{2})^{\otimes m-t}\otimes \chi_{t}\otimes  (\frac{\mathbb{I}+\sigma_Z}{2})^{\otimes N-m-t},
\end{aligned}
\end{equation}
where $\chi_{t}$ is given by
\begin{equation}\label{Eq:DickeXY}
\begin{aligned}
\chi_t=&\sum_{\pi}(\ket{0}\bra{1})^{\otimes t}\otimes(\ket{1}\bra{0})^{\otimes t}\\
=&\sum_{\pi}(\frac{\sigma_X+i\sigma_Y}{2})^{\otimes t}\otimes(\frac{\sigma_X-i\sigma_Y}{2})^{\otimes t}\\
=&\sum_{l=0}^{2t}\sum_\pi \alpha_{t,l}\sigma_X^{\otimes l}\sigma_Y^{\otimes (2t-l)},
\end{aligned}
\end{equation}
where $\alpha_{t,l}$ are the corresponding coefficients, whose values do not affect the following analysis.

Based on  the single qubit operators appearing in $\Theta_t$, we utilize a new orthogonal one-qubit basis  $G_1'=\{\frac{\mathbb{I}-\sigma_Z}{2},\sigma_X,\sigma_Y,\frac{\mathbb{I}+\sigma_Z}{2}\}$ to act as the local basis, and clearly one has $\mathrm{Sym}_N(G_1')=\mathrm{Sym}_N(G_1)$. Similar as Eq.~\eqref{Eq:Mijk}, the corresponding new orthogonal basis of the symmetric subspace is
  \begin{equation}\label{}
\begin{aligned}
M'_{i,j,k}=\sum_{\pi}(\frac{\mathbb{I}-\sigma_Z}{2})^{\otimes i}\otimes\sigma_X^{\otimes j}\otimes\sigma_Y^{\otimes k}\otimes(\frac{\mathbb{I}+\sigma_Z}{2})^{\otimes (N-i-j-k)}.
\end{aligned}
\end{equation}

As a result, $\Theta_t$ lies in the subspace
\begin{equation}\label{}
\begin{aligned}
V_t=\mathrm{span}\{M'_{i,j,k}|i+j+k=m+t\}.
\end{aligned}
\end{equation}
In fact,  $V_t$ is isomorphic to $\mathrm{Sym}_{(m+t)}(\mathbb{R}^3) $, in the sense that
\begin{equation}\label{}
\begin{aligned}
\sum_{\pi}\left(a(\frac{\mathbb{I}-\sigma_Z}{2})+b\sigma_X+c\sigma_Y\right)^{\otimes (m+t)}\otimes (\frac{\mathbb{I}+\sigma_Z}{2})^{\otimes (N-m-t)}=\sum_{i+j+k=m+t}a^ib^jc^kM'_{i,j,k}.
\end{aligned}
\end{equation}

Thus, according to Theorem \ref{Th:basis}, for any two real number sets  $\{b_0,b_1\cdots, b_{m+t}\}$ and $\{c_0, c_1,\cdots,c_{m+t}\}$,
\begin{equation}\label{Eq:BasisV1}
\begin{aligned}
\Big{\{}A_{j,k}^{(t)} =\sum_{\pi}(\frac{\mathbb{I}-\sigma_Z}{2}+b_j\sigma_X+c_k\sigma_Y)^{\otimes (m+t)}\otimes(\frac{\mathbb{I}+\sigma_Z}{2})^{\otimes (N-m-t)} \Big{|}j+k\le m+t\Big{\}}
\end{aligned}
\end{equation}
is a basis set of $V_t$.

Consequently, by constructing $A_{j,k}^{(t)}$ for all $j,k$, each $\Theta_t$ in the $V_t$ subspace can be decomposed. Then after decomposing all $\Theta_t$, the decomposition of $D_{N,m}$ can be obtained. In the following, we show how to construct $A_{j,k}^{(t)}$ with LMSs.

Since $1\leq t \leq m$, the parameter $b_j$, $c_k$ which determine the basis operator $A_{j,k}^{(t)}$ of each subspace $V_t$ in Eq.~\eqref{Eq:BasisV1} are extended to
$\{b_0,b_1\cdots, b_{2m}\}$ and $\{c_0, c_1,\cdots,c_{2m}\}$. That is, we construct $A_{j,k}^{(t)}$ of different $t$ by using the same sets of $b_j$ and $c_k$, which can help us save the number of LMSs.

Deonote $T_{j,k}=\frac{\mathbb{I}-\sigma_Z}{2}+b_j\sigma_X+c_k\sigma_Y$ for simplicity, and the basis operator $A_{j,k}^{(t)}$ of $V_t$ space in Eq.~\eqref{Eq:BasisV1} shows that
\begin{equation}\label{}
\begin{aligned}
A_{j,k}^{(t)}=\sum_{\pi}T_{j,k}^{\otimes (m+t)}(\frac{\mathbb{I}+\sigma_Z}{2})^{\otimes (N-m-t)}.
\end{aligned}
\end{equation}
$A_{j,k}^{(t)}$ is a symmetric operator, generated by the single-qubit operators $T_{j,k}$ and $(\mathbb{I}+\sigma_Z)/2$. Thus, according to Theorem \ref{Th:basis}, for specific $j,k$, one can construct $A_{j,k}^{(t)}$ by the following $N+1$ product basis for any $t$,
\begin{equation}\label{}
\begin{aligned}
\left\{\left(\tan \theta_kT_{i,j}+\frac{\mathbb{I}+\sigma_Z}{2}\right)^{\otimes N}\right\},
\end{aligned}
\end{equation}
where $0\le k\le N$, $0\leq \theta_k < \pi$, and $\theta_k\not=\theta_{k'}$ with $k\neq k'$. For example, $\theta_k=\frac{k\pi}{N+1}$. As a result, after constructing $A_{j,k}^{(t)}$, we can decompose $\Theta_t$, as well as $D_{N,m}$.

Finally, let us count the total number of LMSs. For each $T_{j,k}$, we need $N+1$ LMSs. And there are $\binom{2m+2}{2}=(m+1)(2m+1)$ different $T_{j,k}$'s. Thus there is a total of $(m+1)(2m+1)(N+1)$ LMSs. In fact, one can reduce the number of the LMSs with more careful analysis.  There is one setting,  $(\frac{\mathbb{I}+\sigma_Z}{2})^{\otimes N}$, with $\theta_k=0$ shared by all $T_{j,k}$, which is equivalent to the $\sigma_Z^{\otimes N}$ setting. In addition, if we choose $b_0=c_0=0$, $T_{0,0}$ only needs the same setting $\sigma_Z^{\otimes N}$. As a result, the final number of LMSs is,
\begin{equation}\label{}
\begin{aligned}
((m+1)(2m+1)-1)N+1=m(2m+3)N+1.
\end{aligned}
\end{equation}
\end{proof}

Some remarks are as follows. First, our construction is general and suitable for any Dicke state $\ket{D_{N,m}}$. We expect that the number of LMSs could be reduced further with more elaborate analysis. Specifically, one may find a smaller subspace compared with $V_t$ that contains $\Theta_t$. Second, as $m=1$, our upper bound is $5N+1$, which is larger than the previous result $2N-1$ for the $W$ state. Thus, it is also expected that for certain Dicke state, such as $\ket{D_{N,2}}$, the number of LMSs can also be reduced. And we leave them for further work.
\section{Lower bound on the Measurement Complexity}\label{sec:typical}
In this section, we bound the measurement complexity of the GHZ state, the W state and the Dicke state from below. The corresponding results are listed in the second column of Table \ref{table:MC}.

Before formally discussing the lower bound of the measurement complexity for a specific state, here we explain the high level idea of how to bound it from below.
On the one hand, we find a subspace in which the projection of the state has a certain form; on the other hand, on account of the product form of the LMS, we show that the projection result can only be reconstructed with a high enough number of LMSs.

First, let us study the measurement complexity of the $N$-qubit GHZ state \cite{greenberger1989going},
\begin{equation}\label{}
\begin{aligned}
\ket{GHZ_N}=\frac1{\sqrt{2}}(\ket{0}^{\otimes N}+ \ket{1}^{\otimes N}).
\end{aligned}
\end{equation}
In Ref.~\cite{Guhne2007Toolbox}, the authors decompose $\ket{GHZ_N}$ into $N+1$ LMSs.  In the following theorem, we provide a lower bound for the measurement complexity of the GHZ state which also grows linearly. It means that one should make use of $\Theta(N)$ LMSs to evaluate the fidelity with the GHZ state.
\begin{theorem}\label{Th:lowGHZ}
The measurement complexity of the $N$-qubit GHZ state $\ket{GHZ_N}$ is lower bounded by
\begin{equation}\label{}
C_S(\ket{GHZ_N})\ge \lceil\frac{N+1}{2}\rceil.
\end{equation}
\end{theorem}

\begin{proof}
The density matrix of the GHZ state can be written in the form \cite{Guhne2007Toolbox}
\begin{equation}\label{}
\begin{aligned}
GHZ_N=\frac{1}{2}\left(\frac{\mathbb{I}+\sigma_Z}{2}\right)^{\otimes N} +\frac{1}{2}\left(\frac{\mathbb{I}-\sigma_Z}{2}\right)^{\otimes N} +\frac{1}{2^{N}}\sum^N_{\mathrm{even}\ k=0}(-1)^{k/2}\sum_{\pi}\sigma_Y^{\otimes k}\sigma_X^{\otimes (N-k)},
\end{aligned}
\end{equation}
where the first two terms account for the diagonal elements and the last one for the off-diagonal elements.

Denote the set of LMSs used to decompose the GHZ state as
\begin{equation}\label{eq_setting}
\begin{aligned}
\mathcal{A}=\{A_i^{\otimes N}=(b_i\sigma_X+c_i\sigma_Y+d_i\sigma_Z)^{\otimes N}|i=1,2,\cdots,n_A\},
\end{aligned}
\end{equation}
and the final operator constructed from $\mathcal{A}$ shown in Eq.~\eqref{Eq:MCDef} as $O$,
\begin{equation}\label{Eq:GHZO}
\begin{aligned}
O=\sum_{i=1}^{n_A}\sum_{j=1}^N\alpha_{ij}\sum_{\pi}\mathbb{I}^{\otimes j}A_i^{\otimes N-j}.
\end{aligned}
\end{equation}
It is assumed that $O=GHZ_N$.

%
%
%

Then, consider the projection onto the subspace
\begin{equation}\label{Eq:GHZspace}
\begin{aligned}
\mathrm{span}\{M_{0,N-k,k}|0\le k\le N,\mathrm{even}\ k\},
\end{aligned}
\end{equation}
where $M_{i,j,k}$ is the orthogonal basis defined in (\ref{Eq:Mijk}).

For the GHZ state, we write the projection results on all the basis operators $M_{0,N-k,k}$ in Eq.~\eqref{Eq:GHZspace} in vector form as
\begin{equation}\label{}
\begin{aligned}
v_{GHZ}=\frac{1}{2^N}(1,-1,1,\cdots,(-1)^{\lfloor N/2\rfloor}).
\end{aligned}
\end{equation}
For $O$, it is not hard to see that only the following terms in the summation of Eq.~\eqref{Eq:GHZO} have non-zero projection on this subspace,
\begin{equation}\label{}
\begin{aligned}
\sum_{i=1}^{n_A}\alpha_{i0}A_i^{\otimes N},
\end{aligned}
\end{equation}
and the projection result shows
\begin{equation}\label{Eq:bc0}
\begin{aligned}
v_O=\sum_{i=1}^{n_A}\alpha_{i0}(b_i^N,b_i^{N-2}c_i^2,\cdots ,b_i^{N-2\lfloor N/2\rfloor}c_i^{2\lfloor N/2\rfloor}).
\end{aligned}
\end{equation}

Since $O=GHZ_N$, the projection results should also be equal, i.e., $v_{GHZ}=v_{O}$.
In the following, we show that if $v_G=v_{O}$, the number of LMSs $n_A\ge\lceil\frac{N+1}{2}\rceil$.

Here we focus on the case where $b_i\not=0$ and $c_i\not=0$, and leave the proof of the general case in Appendix \ref{Ap:upGHZ}.
Define $\beta_i=(c_i/b_i)^2$ and it is clear that $\beta_i>0$, then $v_O$ shows
\begin{equation}\label{}
\begin{aligned}
v_O=&\sum_{i=1}^{n_A}\alpha_i(1,\beta_i\cdots,\beta_i^{\lfloor N/2\rfloor}),
\end{aligned}
\end{equation}
where $\alpha_i=\alpha_{i0}b_i^{N}$.

We construct the function
\begin{equation}\label{Eq:GHZgx}
\begin{aligned}
g(x)=\sum_{i=1}^{n_A}\alpha_i\beta_i^x,
\end{aligned}
\end{equation}
which is the linear combination of $n_A$ exponential functions.
Since $v_{O}=v_G$, we have
\begin{equation}\label{}
\begin{aligned}
 g(0)=\frac{1}{2^N},\ g(1)=-\frac{1}{2^N},\ \ g(2)=\frac{1}{2^N},\ \cdots,\ g(\lfloor N/2\rfloor)=(-1)^{\lfloor N/2\rfloor}\frac{1}{2^N}.
\end{aligned}
\end{equation}
and it is clear that $g(x)$ changes its sign with respect to adjacent points.
On account of the continuity of $g(x)$, there is at least one root of $g(x)=0$ in each of the intervals $(0,1),\ (1,2),\ \cdots,\ (\lfloor N/2\rfloor-1,\lfloor N/2\rfloor)$. Totally, there are at least $\lfloor N/2\rfloor$ roots.

On the other hand, it is known that $g(x)$ has $n_A-1$ roots at most \cite{tossavainen2006zeros}, as shown in Lemma \ref{Lem:Exp} below. Consequently, one has $n_A-1\geq\lfloor N/2\rfloor$, i.e., $n_A\geq \lceil\frac{N+1}{2}\rceil$.
\end{proof}

\begin{lemma}\label{Lem:Exp}
\cite{tossavainen2006zeros}
For real numbers $\{\alpha_i\}_1^n$ and $\{\beta_i\}_1^n$ with $\alpha_i \neq 0$, $\beta_i>0$, and $\beta_i\not=\beta_j$ for $i\neq j$, the function
\begin{equation}
\begin{aligned}
g(x)=\sum_{i=1}^{n}\alpha_i\beta_i^x
\end{aligned}
\end{equation}
has at most $n-1$ roots.
\end{lemma}
The proof of Lemma \ref{Lem:Exp} can be found in \cite{tossavainen2006zeros}.

Then, we give the measurement complexity lower bounds of the Dicke state as well as the W state.
\begin{theorem}\label{Th:LowDicke}
The measurement complexity of the Dicke state $\ket{D_{N,m}}$ is lower bounded by
\begin{equation}\label{}
\begin{aligned}
C_S(\ket{D_{N,m}})\ge N-2m+1.
\end{aligned}
\end{equation}
\end{theorem}
As a result of Theorem \ref{Th:LowDicke}, we also have the lower bound for the $W$ state.
\begin{corollary}
The measurement complexity of the $N$-qubit $W$ state $\ket{W_N}$ is lower bounded by
\begin{equation}\label{}
\begin{aligned}
C_S(\ket{W_N})\ge N-1.
\end{aligned}
\end{equation}
\end{corollary}
There is a known decomposition of $\ket{W_N}$ using $2N-1$ LMSs \cite{Guhne2007Toolbox}, which is consistent with our lower bound. This means that one should make use of $\Theta(N)$ LMSs to evaluate the fidelity with the $W$ state.

The proof of Theorem \ref{Th:LowDicke} is similar to that of Theorem \ref{Th:lowGHZ} for the GHZ state. We find a specific subspace where the state $D_{N,m}$ has zero projection. On the other hand, we show that there should be at least $N-2m+1$ LMSs in the decomposition of $D_{N,m}$, in order to make the projection also be zero. The detailed proof is reported in Appendix \ref{Ap:lowDicke}.

\section{conclusion and outlook}\label{sec:final}
In this paper, by introducing a set of product bases for the symmetric subspace, we show that with $(N+1)(N+2)/2$ LMSs, one can decompose a symmetric observable and evaluate the fidelity between an unknown prepared state and any PI state. For some typical PI states, such as the $GHZ$ state, the $W$ state, and the Dicke state with constant number excitations, we can further reduce the measurement complexity down to the linear regime. 

There are a few prospective problems that can be explored in the future. First, it is interesting to show whether the measurement complexity of Dicke states with $\Theta(N)$ excitations, such as $\ket{D_{N,\frac{N}{2}}}$, is still $\Theta(N)$. Second, besides the $GHZ$ state, the $W$ state, and the Dicke state, one might also reduce the measurement complexity for other PI states using similar decomposition techniques. Third, our decomposition technique focuses on the party permutation symmetry but it might also be extended to other types of symmetry, such as the permutation symmetry of eigenstates in high dimensional systems. In addition, the observable decomposition method can be directly applied to entanglement detection, by constructing the corresponding fidelity-based entanglement witnesses \cite{TERHAL2002Detecting,GUHNE2009Entanglement}, where further reduction of the measurement complexity is expected \cite{Guhne2007Toolbox,Qi2019Efficient}. In fact, we find that this kind of construction can yield better witness operators considering entanglement detection under coherent noises \cite{You2019coherent}.

\acknowledgments
Y.~Z.~and C.~G.~contribute equally to this work. We acknowledge Q.~Zhao and X.~Yuan for the insightful discussions. This work was supported by the National Natural Science Foundation of China Grants No.~11875173 and No.~11674193, and the National Key R\&D Program of China Grant No.~2017YFA0303900.

\appendix
\section{Proof of Theorem \ref{Th:basis}}\label{prf:thm1}
Here we give the proof of Theorem \ref{Th:basis}, which constructs a set of product-state bases for the symmetric subspace $\mathrm{Sym}_N(\mathcal{H}_d)$, that is,
\begin{equation}\label{Eq:basisA}
\begin{aligned}
\mathcal{B}=\left\{\ket{\Phi_{\vec{j}}}=\Big(a_{0,j_0}\ket{0}+a_{1,j_1}\ket{1}+\cdots+a_{d-1,j_{d-1}}\ket{d-1}\Big)^{\otimes N}\right\},
\end{aligned}
\end{equation}
where $\sum_{k=0}^{d-1} j_k=N$ and the coefficients $a_{k,j_k}$ are selected from any matrix denoted $A_{d,N}$ satisfying Eq.~\eqref{Eq:Adn}. Since $a_{0,j_0}=1$ for any $j_0$, the basis set in Eq.~\eqref{Eq:basisA} becomes,
\begin{equation}\label{}
\begin{aligned}
\mathcal{B}=\left\{\ket{\Phi_{\vec{j}}}=\Big(\ket{0}+a_{1,j_1}\ket{1}+\cdots+a_{d-1,j_{d-1}}\ket{d-1}\Big)^{\otimes N}\right\},
\end{aligned}
\end{equation}
and the constraint of $j_k$ can be replaced with $\sum_{k=1}^{d-1} j_k\leq N$.

On the other hand, as shown in Eq.~\eqref{Eq:Dbasisd}, there is another set of orthogonal (unnormalized) basis of $\mathrm{Sym}_N(\mathcal{H}_d)$ showing \cite{Harrow2013symmetric},
\begin{equation}\label{}
\begin{aligned}
\left\{\ket{\Psi_{\vec{i}}}=\sum_{\pi}\ket{0}^{\otimes i_0}\ket{1}^{\otimes i_1}\cdots \ket{d-1}^{\otimes i_{d-1}}\right\},
\end{aligned}
\end{equation}
where $\vec{i}=(i_0,i_1,\cdots,i_{d-1})$ a $d$-dimensional vector, with $i_k\in \mathbb{N}$ and $\sum_{k=0}^{d-1} i_k=N$. It is the generalization of Eq.~\eqref{Eq:SymBasis2} from the qubit to the qudit case.

The number of vectors in $\mathcal{B}$ is $D_{S}=\binom{N+d-1}{d-1}=\frac{(N+d-1)!}{N!(d-1)!}$, the dimension of $\mathrm{Sym}_N(\mathcal{H}_d)$. Thus one only needs to show that the vectors in $\mathcal{B}$ are linearly independent. Here, we write $\ket{\Phi_{\vec{j}}}$ in the $\{\ket{\Psi_{\vec{i}}}\}$ basis, and the result shows
\begin{equation}\label{}
\begin{aligned}
\ket{\Phi_{\vec{j}}}=\sum_{\vec{i}}\prod_{k=1}^{d-1}a_{k,j_k}^{i_k}\ket{\Psi_{\vec{i}}}.
\end{aligned}
\end{equation}
The corresponding coefficient matrix of all $\{\ket{\Phi_{\vec{j}}}\}$ shows,
\begin{equation}\label{}
\begin{aligned}
M^{A_{d,n}}_{\vec{i},\vec{j}}=\bra{\Psi_{\vec{i}}}\Phi_{\vec{j}}\rangle=\prod_{k=1}^{d-1}a_{k,j_k}^{i_k},
\end{aligned}
\end{equation}
where each column vector is the coefficient of $\ket{\Phi_{\vec{j}}}$. We say $A_{d,N}$ \emph{generating} $M^{A_{d,N}}$.

In the following we show that the determinant of $M^{A_{d,N}}$ is non-zero, by utilizing the induction method on both $d$ (the local dimension) and $N$ (the qudit number). First, when $d=1$ or $N=1$, it is not hard to check that $\{\ket{\Phi_{\vec{j}}}\}$ are linearly independent and form a basis set.

For general $d$ and $N$, we do the following row transformation on $M^{A_{d,N}}$: for the row index $\vec{i}$ satisfying $i_1=0$, keep these rows unchanged; for the other ones with $1\leq i_1\leq N$, we find the corresponding row with index $\vec{i'}=(i_0+1, i_1-1, i_2, i_3, \cdots, i_{d-1})$, and subtract $a_{1,0}$ multiplying this row $\vec{i'}$. Note that we do this transformation in order from $i_1=N$ to $i_1=1$, and the resulting matrix shows,
\begin{equation}\label{Eq:Mtrans}
M'_{\vec{i},\vec{j}}=\left\{
\begin{aligned}
&\prod_{k=1}^{d-1}a_{k,j_k}^{i_k},&i_1=0\\
&\prod_{k=1}^{d-1}a_{k,j_k}^{i_k}-a_{1,0}\prod_{k=1}^{d-1}a_{k,j_k}^{i_k-\delta(k-1)}\\
&=(a_{1,j_1}-a_{1,0})\prod_{k=1}^{d-1}a_{k,j_k}^{i_k-\delta(k-1)},&1\leq i_1\leq N\\
\end{aligned}
\right.
\end{equation}
where the function $\delta(k-1)=1$ as $k=1$, and otherwise it equals zero.

From Eq.~\eqref{Eq:Mtrans}, one can see that for the matrix elements $M_{\vec{i},\vec{j}}$ satisfying $1\leq i_1\leq N$ and $j_1=0$ equal to zero. That is, $M'_{\vec{i},\vec{j}}$ is an upper triangular block matrix, i.e.,
\begin{equation}
 M'_{\vec{i},\vec{j}}=\left(
 \begin{array}{cc}
    M^1_{\vec{i},\vec{j}} & M_\urcorner\\
    0 & M^2_{\vec{i},\vec{j}}\\
  \end{array}
\right),
\end{equation}
where $M^1_{\vec{i},\vec{j}}$ ($M^2_{\vec{i},\vec{j}}$) is a square matrix which is located in the rows $i_1=(>)0$ and columns $j_1=(>)0$, and $M_\urcorner$ is the off-diagonal part. In the following, we show that the
determinants of $M^1_{\vec{i},\vec{j}}$ and $M^2_{\vec{i},\vec{j}}$ are both non-zero by induction.

Since $i_1=j_1=0$, $M^1_{\vec{i},\vec{j}}$ is generated by the matrix,
\begin{equation}
A^1_{d-1,N}=\left\{a_{k,j}^1=a_{k+1,j}\Big|k\in\{1,\cdots,d-2\},j\in \{0,1,\cdots,N\}\right\}.
\end{equation}
Actually, the matrix $A^1_{d-1,N}$ is related to a $d-1$ local dimension and $N$-fold tensor.
By induction principle, we can get that $\det(M^1_{\vec{i},\vec{j}})\neq 0$.

For the other submatrix $M^2_{\vec{i},\vec{j}}$ with $i_1, j_1\geq1$ shown in Eq.~\eqref{Eq:Mtrans}, the precoefficient $(a_{1,j_1}-a_{1,0})\neq 0$ , and it remains the same in the same column. Since we only care about the non-zero property of $\det (M^2_{\vec{i},\vec{j}})$, we can eliminate these unimportant pre-coefficients and check the determinant of the remaining matrix,
\begin{equation}
M_{\vec{i},\vec{j}}^3=\prod_{k=1}^{d-1}a_{k,j_k}^{i_k-\delta(k-1)}.
\end{equation}
Denoting a new $i_1'=i_1-1$ and $i_k'=i_k$ for $k\neq1$, then one has $\sum_k i_k'=N-1$. In this way, it is not hard to see that $M_{\vec{i},\vec{j}}^3$ can be generated by the matrix,
\begin{equation}
A^3_{d,N-1}=\left\{a_{k,j}^3\Big|k\in\{1,\cdots,d-1\},j\in \{0,1,\cdots,N-1\}\right\},
\end{equation}
where
\begin{equation}\label{}
a_{k,j}^3=\left\{
\begin{aligned}
&a_{k,j+1},k=1,\\
&a_{k,j},k=2,\cdots,d-1.\\
\end{aligned}
\right.
\end{equation}
In fact, $A^3_{d,N-1}$ is related to a $d$ local dimension and $(N-1)$-fold tensor. Again by induction principle, we have that $\det(M^3_{\vec{i},\vec{j}})\neq 0$, and thus $\det (M^2_{\vec{i},\vec{j}})\neq 0$.

Consequently, $\det(M^{A_{d,n}})=\det (M^1_{\vec{i},\vec{j}})\det (M^2_{\vec{i},\vec{j}})\neq 0$, and $\{\ket{\Phi_{\vec{j}}}\}$ is a basis set of the symmetric subspace $\mathrm{Sym}_N(\mathcal{H}_d)$.
\section{Decomposition of the PI state in the product-form basis}\label{Ap:decomp}
In this section, we show explicitly how to efficiently decompose a general PI state in the product-form basis.

The product operators in $\mathcal{B}_o$ in Eq.~\eqref{Eq:basisOp1} form a basis of the operator symmetric subspace $\mathrm{Sym}_N(G_1)$,
\begin{equation}\label{Eq:Ap:Bop}
\begin{aligned}
\mathcal{B}_o=\left\{O_{\vec{\alpha}}=(a_i\mathbb{I}+b_j\sigma_X+c_k\sigma_Y+\sigma_Z)^{\otimes N}\bigg|0\le i,j,k\le N,\; i+j+k\le N\right\}.
\end{aligned}
\end{equation}
where we denote the product operator $O_{\vec{\alpha}}=(a_i\mathbb{I}+b_j\sigma_X+c_k\sigma_Y+\sigma_Z)^{\otimes N}$, with $\vec{\alpha}=\{i,j,k\}$ being a three-dimensional vector as the index. In general, these linearly independent operators $O_{\vec{\alpha}}$ may be not orthogonal. Meanwhile, there is another orthogonal basis of $\mathrm{Sym}_N(G_1)$ shown in Eq.~\eqref{Eq:Oset1},
\begin{equation}\label{}
\begin{aligned}
M_{\vec{\beta}}\doteq M_{i,j,k}
=\sum_{\pi}\mathbb{I}^{\otimes i}\otimes \sigma_X^{\otimes j}\otimes\sigma_Y^{\otimes k}\otimes\sigma_Z^{\otimes (N-i-j-k)}.
\end{aligned}
\end{equation}
where $\vec{\beta}=\{i,j,k\}$ also denotes a three-dimensional vector as the index of $M_{i,j,k}$.

We show in the following how to express a general PI state $\Psi^{PI}$ in the product-form basis $\{O_{\vec{\alpha}}\}$, i.e.,
\begin{equation}\label{}
\begin{aligned}
\Psi^{PI}=\sum_{\vec{\alpha}} \gamma_{\vec{\alpha}} O_{\vec{\alpha}},
\end{aligned}
\end{equation}
where $\gamma_{\vec{\alpha}}$ are real coefficients that we need to figure out.

Our strategy is as follows. First, decompose the PI state on the orthogonal basis $\{M_{\vec{\beta}}\}$ as,
\begin{equation}\label{}
\begin{aligned}
\Psi^{PI}=\sum_{\vec{\beta}} \gamma'_{\vec{\beta}} M_{\vec{\beta}}.
\end{aligned}
\end{equation}
Since the basis operators $M_{\vec{\beta}}$ are orthogonal, the coefficient can be efficiently obtained,
\begin{equation}\label{}
\begin{aligned}
\gamma'_{\vec{\beta}}=c_{\vec{\beta}}\mathrm{Tr}(M_{\vec{\beta}}\Psi^{PI})
\end{aligned}
\end{equation}
where $c_{\vec{\beta}}=\frac{i!j!k!(N-i-j-k)!}{2^NN!}$ is the normalization constant.

Second, do the basis transformation between $M_{\vec{\beta}}$ and $O_{\vec{\alpha}}$. The elements of the basis transformation matrix can be obtained by expressing $O_{\vec{\alpha}}$ on $M_{\vec{\beta}}$, that is,
\begin{equation}\label{}
\begin{aligned}
\Omega_{\vec{\beta},\vec{\alpha}}=c_{\vec{\beta}}\mathrm{Tr}(M_{\vec{\beta}}O_{\vec{\alpha}})
\end{aligned}
\end{equation}

As a result,
\begin{equation}\label{}
\begin{aligned}
\Psi^{PI}&=\sum_{\vec{\alpha}} \gamma_{\vec{\alpha}} O_{\vec{\alpha}}\\
&=\sum_{\vec{\alpha},\vec{\beta}} \gamma_{\vec{\alpha}} \Omega_{\vec{\beta},\vec{\alpha}} M_{\vec{\beta}}\\
&=\sum_{\vec{\beta}} \gamma'_{\vec{\beta}}  M_{\vec{\beta}}.
\end{aligned}
\end{equation}
Thus, one has $\gamma'_{\vec{\beta}}=\sum_{\vec{\alpha}}\Omega_{\vec{\beta},\vec{\alpha}}\gamma_{\vec{\alpha}}$ and $\gamma_{\vec{\alpha}}=\Omega^{-1}\gamma'_{\vec{\beta}}$. Note that the inverse of the matrix $\Omega$ can be evaluated efficiently by the numerical method, since the dimension of the matrix is $\binom{N+2}{2}$.

As a result, by measuring the expectation values of product operators $\langle O_{\vec{\alpha}}\rangle$, one can get the fidelity $\langle \Psi^{PI}\rangle=\sum_{\vec{\alpha}}\gamma_{\vec{\alpha}} \langle O_{\vec{\alpha}}\rangle$ with respective to any PI state $\Psi^{PI}$, by only post-processing the measurement results.

Moreover, as shown in Eq.~\eqref{Eq:Oset1}, one can choose other possible product-form bases by changing the parameters of the local operators that satisfy Eq.~\eqref{Eq:Adn}. Different product-from bases may show different noise tolerances in practical application. For instance, there is some basis choice, where two product operators are too ``close" (even though they are linearly independent), such that they return almost the same result under some measurement imperfection. Thus, we suggest the following selection method, which makes the basis operators be distributed as ``evenly'' as possible. One chooses the coefficients in Eq.~\eqref{Eq:Ap:Bop} as follows. $a_i=\tan \theta_i$ with $\theta_i=\frac{i\pi}{N+1}$, $b_j=\tan \theta_j$ with $\theta_j=\frac{j\pi}{N+1}$, and $c_k=\tan \theta_k$ with $\theta_k=\frac{k\pi}{N+1}$.

Finally, we would like to remark that the method shown above can also be applied to decompose general symmetric operators directly, which may be useful in other related problems.
\section{Proof of Theorem. \ref{Th:lowGHZ} when $b_i=0$ or $c_i=0$ in Eq.~\eqref{Eq:bc0}}\label{Ap:upGHZ}
Here we give the proof of Theorem \ref{Th:lowGHZ} in the case where $b_i=0$ or $c_i=0$ in Eq.~\eqref{Eq:bc0}. Remember that in the main text, we project the operator $GHZ_N$ and $O$  on a specific subspace in Eq.~\eqref{Eq:GHZspace}, and the projection results are, respectively,
\begin{equation}\label{}
\begin{aligned}
v_{GHZ}=\frac{1}{2^N}(1,-1,1,\cdots,(-1)^{\lfloor N/2\rfloor}),
\end{aligned}
\end{equation}
\begin{equation}\label{}
\begin{aligned}
v_O=\sum_{i=1}^{n_A} \alpha_{i0}(b_i^N,b_i^{N-2}c_i^2,\cdots ,b_i^{N-2\lfloor N/2\rfloor}c_i^{2\lfloor N/2\rfloor}).
\end{aligned}
\end{equation}
Then we define a function $g(x)$ in Eq.~\eqref{Eq:GHZgx} which is a summation of several exponential functions, and use the root property of it to bound the number of LMSs.

Here we consider the general case where $b_i=0$ or $c_i=0$. Let $S$ denote the set of $i$ with both $b_i$ and $c_i$ not equal to $0$, and define $\beta_i=(c_i/b_i)^2$ only on $S$. Let $S_b$ denote the set of $i$ with $c_i=0$, $b_i\not=0$ and let $S_c$ denote the set of $i$ with $b_i=0$, $c_i\not=0$. Then $v_O$ can be written as
\begin{equation}\label{}
\begin{aligned}
v_O=&\sum_{i\in S}\alpha_i(1,\beta_i\cdots,\beta_i^{\lfloor N/2\rfloor})+(\alpha_b,0,\cdots ,0)+(0,\cdots ,0,\alpha_c),
\end{aligned}
\end{equation}
where $\alpha_i=\alpha_{i0}b_i^{N}$, $\alpha_b=\sum_{i\in S_b}\alpha_{i0}b_i^{N}$ and $\alpha_c=\sum_{i\in S_c}\alpha_{i0}b_i^{N-2\lfloor N/2\rfloor}c_i^{2\lfloor N/2\rfloor}$. It is clear that

\begin{equation}\label{Eq:SetCount}
\begin{aligned}
n_A\geq |S|+|S_b|+|S_c|.
\end{aligned}
\end{equation}

And we define $g(x)$ of the set $S$ as,
\begin{equation}\label{}
\begin{aligned}
g(x)=\sum_{i\in S}\alpha_i\beta_i^x.
\end{aligned}
\end{equation}

Since $v_O=v_G$, $g(0)=1-\alpha_b,\ g(1)=-1,\ \cdots,\ g(\lfloor N/2\rfloor)=(-1)^{\lfloor N/2\rfloor}-\alpha_c$.
Because the continuity of $g(x)$, there is at least one root of $g(x)$ in each of the intervals $(1,2),\ (2,3),\ \cdots,\ (\lfloor N/2\rfloor-2,\lfloor N/2\rfloor-1)$, which means $\lfloor N/2\rfloor-2$ roots in total.
If $\alpha_b=0$,  there is at least another root in $(0,1)$. Similarly, if $\alpha_c=0$,  there is at least another root in $(\lfloor N/2\rfloor-1,\lfloor N/2\rfloor)$.

Therefore, the number of roots of $g(x)$ is at least
\begin{equation}\label{}
\begin{aligned}
&\lfloor N/2\rfloor-2+I(\alpha_b=0)+I(\alpha_c=0)\\
=& \lfloor N/2\rfloor-I(\alpha_b\not=0)-I(\alpha_c\not=0)\\
 \ge&\lfloor N/2\rfloor-|S_b|-|S_c|,
\end{aligned}
\end{equation}
where $I(x)$ denote that function that $I(x)=0$ or $1$ when $x$ is ture or false.

Then apply Lemma \ref{Lem:Exp} to $g(x)$, we get $|S|-1\ge \lfloor N/2\rfloor-|S_b|-|S_c|$. Combining this with \eqref{Eq:SetCount}, one has $n_A\ge \lceil\frac{N+1}{2}\rceil$.
\section{Proof of Theorem \ref{Th:LowDicke}}\label{Ap:lowDicke}
Here we give the proof of Theorem \ref{Th:LowDicke}. As mentioned in the text, we find a subspace where $D_{N,m}$ has zero projection, and at the same time show that one needs at least $N-2m+1$ LMSs to make the projection also be zero.

As in the GHZ state case, suppose the optimal LMSs are
\begin{equation}\label{Eq:LMSAp}
\begin{aligned}
\mathcal{A}=\{A_i^{\otimes N}=(b_i\sigma_X+c_i\sigma_Y+d_i\sigma_Z)^{\otimes N}|i=1,2,\cdots,n_A\}.
\end{aligned}
\end{equation}
The final operator constructed from $\mathcal{A}$ is denoted $O$,
\begin{equation}\label{}
\begin{aligned}
O=\sum_{i=1}^{n_A}\sum_{j=1}^N\alpha_{ij}\sum_{\pi}\mathbb{I}^{\otimes j}A_i^{\otimes N-j},
\end{aligned}
\end{equation}
and it is assumed that $O=D_{N,m}$.

As shown in Eqs.~\eqref{Eq:DecomDicke}, \eqref{Eq:Theta} and \eqref{Eq:DickeXY},  $D_{N,m}$ can be decomposed as,
\begin{equation}\label{Eq:App:D}
\begin{aligned}
D_{N,m}=\frac{1}{\binom{N}{m}}\sum_{t=0}^{m}\Theta_t,\\
\end{aligned}
\end{equation}
with
\begin{equation}\label{}
\begin{aligned}
&\Theta_t=\sum_{\pi}(\frac{\mathbb{I}-\sigma_Z}{2})^{\otimes m-t}\otimes \chi_{t}\otimes  (\frac{\mathbb{I}+\sigma_Z}{2})^{\otimes N-m-t},\\
&\chi_t=\sum_{l=0}^{2t}\sum_\pi \alpha_{t,l}\sigma_X^{\otimes l}\sigma_Y^{\otimes (2t-l)}.
\end{aligned}
\end{equation}
It is not hard to see that the total number of $\sigma_X$ and $\sigma_Y$ operators on $N$ qubits is at most $2m$, appearing in every terms of the decomposition in Eq.~\eqref{Eq:App:D}.  Thus $D_{N,m}$ lies in the subspace, $\mathrm{span}\{M_{i,j,k}|0\le j+k\le 2m\}$, where $M_{i,j,k}$ are defined in (\ref{Eq:Mijk}).

For the constructed operator $O$, there should be some operator $A_i$ satisfying $b_i\neq0$ in the LMSs in Eq.~\eqref{Eq:LMSAp}, otherwise there will be no $\sigma_X$ term in $O$. If all the operator $A_i$ with $b_i\neq0$ satisfy $d_i=0$, $O\not=D_{N,m}$, since there are terms like $\sum_\pi \sigma_X^{\otimes 2t} \sigma_Z^{\otimes N-2t} $ in $D_{N,m}$.

Now consider the subspace
\begin{equation}\label{Eq:Ap:Mj}
\begin{aligned}
V_1=\mathrm{span}\{M_{0,j,0}|2m+1\le j\le N\},
\end{aligned}
\end{equation}
where $D_{N,m}$ has zero component on. For $O$, only the following terms could have non-zero projection on this subspace,
\begin{equation}\label{}
\begin{aligned}
\sum_{i=1}^{n_A}\alpha_{i0}A_i^{\otimes N},
\end{aligned}
\end{equation}
and we write the projection on $M_{0,j,0}$ in Eq.~\eqref{Eq:Ap:Mj} in vector form as,
\begin{equation}\label{Eq:ODicke1}
\begin{aligned}
v_O=\sum_{i=1}^{n_A}\alpha_{i0}(b_i^{2m+1}d_i^{N-2m-1},b_i^{2m+2}d_i^{N-2m-2},\cdots,b_i^N),
\end{aligned}
\end{equation}
where we consider $b_i\neq0$ otherwise it contribute nothing to the summation.  On account of $O=D_{N,m}$, this projection result should also be zero, i.e., $v_O=\vec{0}$.

First, we focus on the case where all $d_i\neq0$, denote $\beta_i=b_i/d_i$. And $v_O$ can be written as,
\begin{equation}\label{Eq:ODicke}
\begin{aligned}
v_O=\sum_i \alpha_{i}(1,\beta_i,\cdots,\beta_i^{N-2m-1})^T=\vec{0}.
\end{aligned}
\end{equation}
where $\alpha_{i}$ is the summation of all the corresponding coefficients sharing the same $\beta_i$,
\begin{equation}\label{}
\begin{aligned}
\alpha_{i}=\sum_{i':\beta_{i'}=\beta_i}\alpha_{i'0}b_{i'}^{2m+1}d_{i'}^{N-2m-1},
\end{aligned}
\end{equation}

In fact, there is at least one $\alpha_{i}\neq0$. To illustrate this, let us consider another subspace,
\begin{equation}\label{}
\begin{aligned}
V_2=\mathrm{span}\{M_{0,j,0}|1\le j\le 2m\}.
\end{aligned}
\end{equation}
It is clear that the projection of $D_{N,m}$ on $V_2$ is nonzero. For example, the terms like $\sum_\pi  \sigma_X^{\otimes 2t}\sigma_Z^{\otimes N-2t}$ are the basis of it. In the meantime, the projection of $O$ on $V_2$ is,
\begin{equation}\label{}
\begin{aligned}
v_O'=\sum_i \alpha_{i}(\beta_i^{-2m},\beta_i^{-2m+1}\cdots,\beta_i^{-1})^T.
\end{aligned}
\end{equation}
Consequently, there is at least one $\alpha_{i}\neq0$, otherwise $v_O'=\vec{0}$, which is in contradiction to $O=D_{N,m}$.

Denote the number of different $\beta_i$ as $n_{\beta}$. Then Eq.~\eqref{Eq:ODicke} means that an $(N-2m)\times n_{\beta}$ Vandermonde matrix multiplies a non-zero vector $\{\alpha_i\}$. Due to the non-singularity of the Vandermonde matrix, the result can be $\vec{0}$, only if $n_{\beta}>N-2m$. As a result, the number of measurement settings is lower bounded by $n_A\geq n_{\beta}>N-2m$.

For the case where there are LMSs with $d_i=0$, denote the set of these LMSs as $S$, the projection in Eq.~\eqref{Eq:ODicke1} shows
\begin{equation}\label{}
\begin{aligned}
&\sum_{i\in S}\alpha_{i0}b_j^N(0,0,\cdots, 1)^T+\sum_{i\in [n_A]\backslash S}\alpha_{i0}b_j^{2m+1}d_i^{N-2m-1}(1,\beta_i,\cdots,\beta_i^{N-2m-1})^T=\vec{0}.
\end{aligned}
\end{equation}
Denote $\sum_{i\in S}\alpha_{i0}b_j^N=\alpha'$. If $\alpha'\neq0$, it just adds one vector in the linear combination compared with Eq.~\eqref{Eq:ODicke},
\begin{equation}\label{}
\begin{aligned}
\alpha'(0,0,\cdots, 1)^T+\sum_i \alpha_{i}(1,\beta_i,\cdots,\beta_i^{N-2m-1})^T=\vec{0}.
\end{aligned}
\end{equation}
Similarly, based on the non-singularity of the Vandermonde matrix, $n_{\beta}+1>N-2m$. Hence, $n_A\geq n_{\beta}+1>N-2m$.

\bibliographystyle{apsrev4-1}

\bibliography{BibSym}

\end{document}